\documentclass[runningheads]{llncs}

\usepackage{amssymb}
\usepackage{amsmath}
\usepackage{booktabs}
\usepackage{graphicx}
\usepackage{xspace}
\usepackage{hyperref}
\usepackage{todonotes}
\usepackage[font=small,labelfont=bf]{caption}
\usepackage[subrefformat=parens,labelfont=default,font=small]{subcaption}
\usepackage[misc,geometry]{ifsym} 
\usepackage{wrapfig}

\title{Variants of the Segment Number of a Graph\thanks{A.W.\
    acknowledges support from DFG grant WO$\,$758/9-1.}}

\author{%
  Yoshio Okamoto\inst{1}\lncsarxiv{\orcidID{0000-0002-9826-7074}}{}%
  \and%
  Alexander Ravsky\inst{2}
  \and%
  Alexander~Wolff\inst{3}\lncsarxiv{\orcidID{0000-0001-5872-718X}$^{\textrm{(\Letter)}}$}{}
}

\authorrunning{Y.~Okamoto et al.}

\institute{University of Electro-Communications, Ch\=ofu, Japan \\
  and RIKEN Center for Advanced Intelligence Project, Tokyo, Japan%
  \and%
  Pidstryhach Institute for Applied Problems
  of Mechanics and Mathematics, National Academy of Sciences of
  Ukraine, Lviv, Ukraine \\
  \email{alexander.ravsky@uni-wuerzburg.de}
  \and%
  Universit\"at W\"urzburg, W\"urzburg, Germany \lncsarxiv{\\
  \email{usetheemailaddressonmyhomepage@gmail.com}}{}
}

\newcommand{\lncsarxiv}[2]{#2}

\graphicspath{{figures/}}

\newtheorem{open}{Open Problem}

\DeclareMathOperator{\seg}{seg}
\newcommand{\segx}{\ensuremath{\seg_{\hspace{-.1ex}\times\!}}\xspace}
\newcommand{\segbend}{\ensuremath{\seg_\angle}\xspace}
\newcommand{\segtwo}{\ensuremath{\seg_2}\xspace}
\newcommand{\segthree}{\ensuremath{\seg_3}\xspace}
\newcommand{\T}{\ensuremath{\mathcal{T}}}
\renewcommand{\S}{\ensuremath{\mathcal{S}}}
\DeclareMathOperator{\slope}{slope}
\newcommand{\AGR}{\textsc{Arrangement Graph Recognition}\xspace}
\newcommand{\PSTR}{\textsc{Pseudoline Stretchability}\xspace}
\newcommand{\ETR}{\ensuremath{\exists\mathbb{R}}\xspace}
\newcommand{\E}{\ensuremath{\exists}}

\let\doendproof\endproof
\renewcommand\endproof{~\hfill$\qed$\doendproof}

\begin{document}

\maketitle

\begin{abstract}
  The \emph{segment number} of a planar graph is the smallest number
  of line segments whose union represents a crossing-free straight-line
  drawing of the given graph in the plane.  The segment number is a
  measure for the visual complexity of a drawing; it has been studied
  extensively.

  In this paper, we study three variants of the segment number: for
  planar graphs, we consider crossing-free polyline drawings in~2D;
  for arbitrary graphs, we consider crossing-free straight-line
  drawings in~3D and straight-line drawings with crossings in~2D.
  We first construct an infinite family of planar graphs where the
  classical segment number is asymptotically twice as large as each of
  the new variants of the segment number.  Then we establish the
  $\exists\mathbb{R}$-completeness (which implies the NP-hardness) of
  all variants.  Finally, for cubic graphs, we prove lower and upper
  bounds on the new variants of the segment number, depending on the
  connectivity of the given graph.
\end{abstract}

\section{Introduction}

When drawing a graph, a way to keep the visual complexity low is to
use few geometric objects for drawing the edges.  This idea is
captured by the \emph{segment number} of a (planar) graph, that is,
the smallest number of crossing-free line segments that together
constitute a straight-line drawing of the given graph.
The \emph{arc number} of a graph is defined analogously with respect
to circular-arc drawings.
So far, both numbers have only been studied for planar graphs.
Two obvious lower bounds for the segment number are
known~\cite{desw-dpgfs-CGTA07}: (i)~$\eta(G)/2$, where $\eta(G)$ is
the number of odd-degree vertices of~$G$, and (ii)~the planar slope
number of~$G$, that is, the smallest number~$k$ such that $G$ admits a
crossing-free straight-line drawing whose edges have $k$ different
slopes.

Dujmovi\'c et al.~\cite{desw-dpgfs-CGTA07}, who introduced segment
number and planar slope number, showed among others that trees can be
drawn without crossings such that the optimum segment number and the
optimum planar slope number are achieved simultaneously.  In fact,
any tree~$T$ admits a drawing with $\eta(T)/2$ segments and
$\Delta(T)/2$ slopes, where $\Delta(T)$ is the maximum degree
of~$T$.  Unfortunately, these drawings need exponential area.  
Therefore, Schulz~\cite{s-dgfa-JGAA15} suggested to
study the arc number of planar graphs.  Among other things, he showed
that any $n$-vertex tree can be drawn on a polynomial-size grid
($O(n^{1.81}) \times n$) using at most $3n/4$ arcs.

Another measure for the visual complexity of a drawing of a graph is
the minimum number of \emph{lines} whose union contains a straight-line
crossing-free drawing of the given graph.  This parameter is called
the \emph{line cover number} of a graph~$G$ and denoted by
$\rho^1_2(G)$ for~2D (where $G$ must be planar) and $\rho^1_3(G)$
for~3D.  Together with the plane cover number $\rho^2_3(G)$ and other
variants, these parameters have been introduced by Chaplick et
al.~\cite{ChaplickFLRVW16}.  They also showed that both line cover
numbers are \ETR-hard to compute \cite{ChaplickFLRVW16b}.  (For
background on \ETR, see Schaefer's work~\cite{s-cgtp-GD09}.)

Upper bounds for the segment number and the arc number (in terms of
the number of vertices, $n$, ignoring constant additive terms) are
known for series-parallel graphs ($3n/2$ vs.\ $n$), planar 3-trees
($2n$ vs.\ $11n/6$), and triconnected planar graphs ($5n/2$ vs.\ $2n$)
\cite{desw-dpgfs-CGTA07,s-dgfa-JGAA15}.  The upper bound on the
segment number for triconnected planar graphs has been improved for
the special cases of triangulations and 4-connected triangulations
(from $5n/2$ to $7n/3$ and $9n/4$, respectively) by Durocher and
Mondal~\cite{dm-dptfs-CCCG14}.  For the special case of triconnected
cubic graphs, Dujmovi\'c et al.~\cite{desw-dpgfs-CGTA07} showed that 
the segment number is upperbounded by $n+2$.  (A cubic graph with $n$
vertices has $3n/2$ edges.)  The result of Dujmovi\'c et al.\ was
improved by Mondal et al.~\cite{mnbr-mscd3-JCO13} who gave two
linear-time algorithms based on cannonical decompositions; one that
uses at most $n/2+3$ segments for $n \ge 6$ and one that uses $n/2+4$
segments but places all vertices on a grid of size $n \times n$.  Both
algorithms use at most six different slopes.  Note that $n/2+3$
segments are optimal for cubic planar graphs since in every vertex at
least one segment must end and in the at least three vertices on the
convext hull all three incident segments must end.  Igamberdiev et
al.~\cite{ims-dpc3c-JGAA17} fixed a bug in the algorithm of Mondal et
al., presented two conceptually different (but slower) algorithms that
meet the lower bound and compared them experimentally in terms of
common metrics such as angular resolution.

H{\"u}ltenschmidt et al.~\cite{hkms-dttfg-WG17}
provided bounds for segment and arc number under the additional
constraint that vertices must lie on a polynomial-size grid.  They
also showed that $n$-vertex triangulations can be drawn with at most
$5n/3$ arcs, which is better than the lower bound of $2n$ for the
segment number on this class of graphs.  For 4-connected
triangulations, they need at most $3n/2$ arcs.  Kindermann et
al.~\cite{kmss-dpgfs-GD19} recently strengthened some of these
results by showing that many classes of planar graphs admit
nontrivial bounds on the segment number even when restricting
vertices to a grid of size $O(n) \times O(n^2)$.  For drawing
$n$-vertex trees with at most $3n/4$ segments, they reduced the grid
size to $n \times n$.  Among other things, Durocher et
al.~\cite{dmnw-nmsdp-JGAA13} showed that the segment number is NP-hard
to compute \emph{with respect to a fixed embedding}, even in the
special case of arrangement graphs.  They also showed that the following
partial representation extension problem is NP-hard: given an
outerplanar graph~$G$, an integer~$k$, and a straight-line
drawing~$\delta$ of a subgraph of~$G$, is there a $k$-segment drawing
that contains~$\delta$?  It is still open, however, whether the
segment number is fixed-parameter tractable.

In this paper, we consider several variants of the planar segment
number~\segtwo that has been studied extensively.  In particular, we
study the \emph{3D segment number} \segthree, which is the most obvious
generalization of the planar segment number.  It is the smallest
number of straight-line segments needed for a crossing-free
straight-line drawing of a given graph in 3D.  We also study the
\emph{crossing segment number} \segx in 3D, where edges are allowed to
cross, but they are not allowed to overlap or to contain vertices in their
interiors. In this case, by Lemma~\ref{lem:projection}, the minimum number of segments
constituting a drawing of a given graph can be achieved by a plane drawing.
 Finally, for planar graphs, we study the \emph{bend segment
  number} \segbend in 2D, which is the smallest number of straight-line
segments needed for a crossing-free polyline drawing of a given graph
in 2D.

Durocher et al.~\cite{dmnw-nmsdp-JGAA13} were also interested in the
3D segment number.  They stated that their proof of the NP-hardness of
the above-mentioned partial representation problem can be adjusted
to~3D.  They suspected that the 3D segment number remains NP-hard to
compute even if the given graph is subcubic.  Instead, they showed
that a variant of the 3D segment number is NP-hard where one is given a
3D drawing and additional co-planarity constraints that must be
fulfilled in the final drawing.

\paragraph{Our Contribution.}
First, we establish some relationships between the variants of the
segment number; see Section~\ref{sec:relationships}.  Then we turn to
the complexity of computing the new variants of the segment number;
see Section~\ref{sec:hardness}.  By re-using ideas from the
\ETR-completeness proof of Chaplick et al.~\cite{ChaplickFLRVW16b}
regarding the computation of the line cover numbers~$\rho^1_2$
and~$\rho^1_3$, we establish the \ETR-completeness (and hence the
NP-hardness) of all variants of the segment number~--
\segtwo, \segthree, \segx, and \segbend~-- even for
graphs of maximum degree~$4$.  Thus, we nearly answer the open problem of Durocher
et al.~\cite{dmnw-nmsdp-JGAA13} concerning the computational
complexity of the 3D segment number for subcubic graphs.
Note that Hoffmann~\cite{Hoffmann17} recently established the 
\ETR-hardness of computing the slope number $\slope(G)$ of a 
planar graph $G$. 

Our main contribution consists in algorithms and lower-bound
constructions for connected ($\gamma=1$), biconnected ($\gamma=2$),
and triconnected ($\gamma=3$) cubic graphs; see 
Table~\ref{tab:results}.  To put these results into perspective, recall
that any cubic graph with $n$ vertices needs at least $n/2+3$ and at
most $3n/2$ segments to be drawn, regardless of the drawing style.
(In contrast, four slopes slopes suffice for cubic graphs
\cite{ms-grcg4-CGTA09}).  We prove our bounds in
Section~\ref{sec:cubic}.  Note that for cubic graphs, vertex- and
edge-connectivity are the same \cite[Thm.~2.17]{cz-cgt-08}.

\begin{table}[tb]
  \centering

  \caption{Overview over existing and new bounds on variants of the
    segment number of cubic graphs.  The upper bounds hold for all
    $n$-vertex graphs of a certain vertex connectivity~$\gamma$.
    The lower bounds are existential; there exist graphs for which
    they hold.  Note that \segtwo and \segbend are defined only for
    planar graphs.  We skip more specialized known results (e.g.,
    concerning grid size~\cite{hkms-dttfg-WG17}
    or triangulations~\cite{dm-dptfs-CCCG14}).}
  \label{tab:results}

  \medskip

  \newcommand{\mysp}{\hspace{3.5ex}}
  \begin{tabular}{@{}l@{\mysp}ll@{\mysp}ll@{\mysp}ll@{\mysp}ll@{}}
    \toprule
    $\gamma$ & \multicolumn{2}{c}{$\segtwo(G)$}
    & \multicolumn{2}{c}{$\segthree(G)$}
    & \multicolumn{2}{c}{$\segbend(G)$}
    & \multicolumn{2}{c}{$\segx(G)$} \\
    \midrule
    1 & $\ge 5n/6$ & [Prp.~\ref{prop:5n-over-6-example}]
      & $\ge 5n/6$ & [Prp.~\ref{prop:5n-over-6-example}]
      & $\ge 5n/6$ & [Prp.~\ref{prop:5n-over-6-example}]
      & $\ge 5n/6$ & [Prp.~\ref{prop:5n-over-6-example}] \\
    2 & & %
      & $\le n+2$  & [Th.~\ref{thm:bi-cubic-3d}]
      & $\le n+1$  & [Th.~\ref{thm:planar-bi-cubic-bend}]
      & $\le n+2$  & [Th.~\ref{thm:bi-cubic-3d}] \\
      & $\ge 3n/4$ & [Prp.~\ref{prop:claw-cycle}]
      & $\ge 5n/6$ & [Prp.~\ref{prop:K33-cycle}]
      & $\ge 3n/4$ & [Prp.~\ref{prop:claw-cycle}]
      & $\ge 3n/4$ & [Prp.~\ref{prop:claw-cycle}] \\
    3 & $=  n/2+3$ & \cite{ims-dpc3c-JGAA17,mnbr-mscd3-JCO13}
      & $\le n+2$  & [Th.~\ref{thm:bi-cubic-3d}]
      &&& $\le n+2$& [Th.~\ref{thm:bi-cubic-3d}] \\
      & \multicolumn{2}{l@{\mysp}}{\hspace*{-1.5ex}(except for $G=K_4$)} 
      & $\ge 7n/10$ & [Prp.~\ref{prop:tri-cubic-3d}]
      & \multicolumn{2}{c}{$\segbend\equiv\segtwo$} & \\
    \bottomrule
  \end{tabular}
\end{table}

Before we start, we introduce the following notation.
For a given polyline drawing~$\delta$ of a graph in~2D or~3D,
we denote by~$\seg(\delta)$ the number of (inclusionwise maximal)
straight-line segments of which the drawing~$\delta$ consists.

\section{Relationships Between Segment Number Variants}
\label{sec:relationships}

\begin{lemma}
  \label{lem:projection}
  Given a graph $G$ and a straight-line drawing $\delta$ of~$G$ in 3d
  with the property that no two edges overlap and no edge contains a
  vertex in its interior, then there exists a plane drawing~$\delta'$
  of~$G$ with $\seg(\delta') \le \seg(\delta)$ and with the same
  property as~$\delta$.  (Note that both in $\delta$ and $\delta'$
  edges may cross.)
\end{lemma}

\begin{proof}
  For each triplet $u,v,w$ of points in~$\delta$ that correspond to 
  three distinct vertices of~$G$, let $P(u,v,w)$ be the
  plane or line spanned by the vectors $\overrightarrow{uv}$ and
  $\overrightarrow{wv}$, and
  let~$\mathcal P$ be the set of all such planes or lines.
  Choose a point~$A$ in $\mathbb{R}^3 \setminus \bigcup \mathcal{P}$
  that does not lie in the xy-plane.  Let~$\delta'$ be the drawing
  that results from projecting~$\delta$ parallel to the vector $OA$
  onto the xy-plane.  Due to the choice of the projection, $\delta'$
  may contain crossings, but no edge contains a vertex to which it is
  not incident, and no two edges overlap. By construction,
  $\seg(\delta') \le \seg(\delta)$.
\end{proof}

\begin{corollary}
  \label{cor:projection}
  For any graph~$G$ it holds that $\segx(G) \le \segthree(G)$.
\end{corollary}

\begin{proposition}
  \label{prop:T+fans}
  There is an infinite family of planar graphs $(\S_i)_{i\ge3}$ such
  that $\S_i$ has $n_i=i^3-i+6$ vertices and the ratios
  $\segtwo(\S_i)/\segthree(\S_i)$, $\segtwo(\S_i)/\segbend(\S_i)$, and
  $\segtwo(\S_i)/\segx(\S_i)$ all converge to~2 with increasing~$i$.
\end{proposition}

\begin{proof}
  We construct, for $i \ge 3$, a triangulation~$\T_i$ with
  maximum degree~6 and~$t_i=i^2-2i+3$ vertices (and, hence, $3t_i-6$
  edges and $2t_i-4$ faces), as follows.  Take two triangular grids of
  side length~$i-1$ (a single triangle is a grid of side length~1)
  and glue their boundaries, identifying corresponding
  vertices and edges.  Clearly, the result is a (planar) triangulation.
  Let~$s_i = \segtwo(\T_i)$.  Then, by the result of
  Dujmovi\'c et al.~\cite{desw-dpgfs-CGTA07}, $s_i \le 5t_i/2$.

  We assume that $i$ is even.  To each vertex~$v$
  of the triangulation, we attach an $i$-\emph{fan}, that is, a path
  of length~$i$ each of whose vertices is connected to~$v$.
  Let~$\S_i$ be the resulting graph, which has $n_i=t_i(i+2)$ vertices.

  \begin{figure}[tb]
    \begin{subfigure}[b]{.32\textwidth}
      \centering
      \includegraphics[page=1]{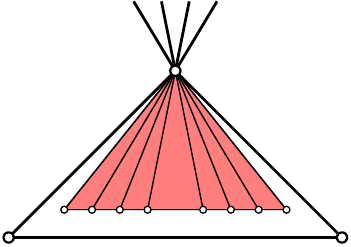}
      \caption{planar case}
      \label{fig:triangle+fan}
    \end{subfigure}
    \hfill
    \begin{subfigure}[b]{.32\textwidth}
      \centering
      \includegraphics[page=2]{triangle+fan}
      \caption{with crossings or in 3D}
      \label{fig:crossing-fan}
    \end{subfigure}
    \hfill
    \begin{subfigure}[b]{.32\textwidth}
      \centering
      \includegraphics[page=3]{triangle+fan}
      \caption{with bends}
      \label{fig:bend-fan}
    \end{subfigure}
    \caption{Attaching a fan (thin edges) to a vertex of a
      triangulation (thick edges) of maximum degree~6}
    \label{fig:triangulation}
  \end{figure}

  In 2D, no matter how the triangulation is drawn, only three vertices
  lie on the outer face.  Consider an $i$-fan incident to one of the
  $t_i-3$ inner vertices; see Fig.~\ref{fig:triangle+fan}.  Each such
  $i$-fan must be placed into a triangular face and needs at least
  $i-3$ segments that are disjoint from the drawing of the
  triangulation.  (Here we use that every vertex has degree at
  most~6.)  Hence, $\segtwo(\S_i) \ge (t_i-3)\cdot(i-3)=i^3-O(i^2)$.

  In 3D on the other hand, we can draw every fan in a plane different
  from the triangulation such that the fan's path lies on three
  segments and the remaining edges are paired such that each pair
  shares a segment; see Fig.~\ref{fig:crossing-fan}.  Hence,
  $\segthree(\S_i) \le t_i \cdot (i/2+3) + s_i = i^3/2 + O(i^2)$.  Due
  to Corollary~\ref{cor:projection}, $\segx(\S_i) \le
  \segthree(\S_i)$.

  To bound $\segbend(\S_i)$, observe that we can modify the layout of
  the triangulation as in Fig.~\ref{fig:bend-fan} such that every
  vertex is incident to an angle greater than~$\pi$ without any incoming
  edges.  This can be achieved as follows.  On each inner vertex~$v$,
  place a disk~$D_v$ whose radius is (slightly smaller than) the
  minimum over the lengths of the incident edges divided by~2 and over
  the distances to all non-incident edges.  The resulting disks have
  positive radii and are pairwise disjoint.  Now we go through all
  vertices.  Let~$v$ be the current vertex and let $\partial D_v$ be
  the boundary of~$D_v$.  We bend all edges incident to~$v$
  at~$\partial D_v$ and place~$v$ on some unused point on~$\partial
  D_v$.  As a result, every vertex is incident to an angle
  greater than~$\pi$ without any incoming edges.  In this area (marked red
  in Fig.~\ref{fig:bend-fan}), we can place the corresponding fan.
  The modification introduces at most two bends in every edge of the
  triangulation.  Hence,
  $\segbend(\S_i) \le t_i \cdot (i/2+3) + 3 \cdot (3t_i-6) =
  i^3/2+O(i^2)$.
\end{proof}

\begin{open} \sloppy
  What are upper bounds for the ratios $\segtwo(G)/\segthree(G)$,
  $\segtwo(G)/\segbend(G)$, and $\segtwo(G)/\segx(G)$ with $G$ ranging
  over all planar graphs?
\end{open}

\section{Computational Complexity}
\label{sec:hardness}

Chaplick et al.~\cite[Theorem~1]{ChaplickFLRVW16b} showed that it is
\ETR-hard to decide for a planar graph~$G$ and an integer~$k$ whether
$\rho^1_2(G) \le k$ and whether $\rho^1_3(G) \le k$.
We follow their approach to show the hardness of all
variants of the segment number that we study in this paper.

A \emph{simple line arrangement} is a set $\mathcal{L}$ of $k$ lines
in $\mathbb R^2$ such that each pair of lines has one intersection
point and no three lines share a common point. We define the
\emph{arrangement graph} for a set of lines as
follows~\cite{BoseEW03}: The vertices correspond to the intersection
points of lines and two vertices are adjacent in the graph if and only
if they lie on the same line and no other vertex lies between them.
The \AGR problem is to decide whether a given graph is the arrangement
graph of some set of lines.

Bose et al.~\cite{BoseEW03} showed that this problem is NP-hard by
reduction from a version of \PSTR for the Euclidean plane, whose
NP-hardness was proved by Shor \cite{Shor91}. It turns out that \AGR
is actually an \ETR-complete problem \cite[page
212]{Eppstein14}. This stronger statement follows from the fact that
the Euclidean \PSTR is \ETR-hard as well as the original
projective version \cite{Matousek14,s-cgtp-GD09}.

\begin{theorem}
  \label{thm:hard}
  Given a planar graph $G$ of maximum degree $4$ and an integer $k$, it is \ETR-hard
  to decide whether $\segtwo(G) \le k$, whether $\segbend(G) \le k$,
  and whether $\segx(G) \le k$.
\end{theorem}

\begin{proof}
  Similarly to Chaplick et al.~\cite[proof of Theorem~1]{ChaplickFLRVW16b},
  we first observe that if $G$ is an arrangement graph, there must be an
  integer~$\ell$ such that $G$ has $\ell(\ell-1)/2$ vertices (of
  degree $d\in\{2,3,4\}$) and $\ell(\ell-2)$ edges.  This uniquely
  determines~$\ell$.  We set the parameter~$k$ from the statement of
  our theorem to this value of~$\ell$.
  Again, as Chaplick et al., we construct a
  graph~$G'$ from $G$ by appending a tail (i.e., a degree-1 vertex) to
  each degree-3 vertex of $G$ and two tails to each degree-2 vertex
  of~$G$.

We claim that the following five conditions are equivalent:
(i)~$G$ is an arrangement graph on $k$ lines,
(ii)~$\rho^1_2(G')\le k$, %
(iii)~$\segtwo(G')\le k$, %
(iv)~$\segbend(G')\le k$, and %
(v)~$\segx(G')\le k$. %
Once the equivalence is established, the \ETR-hardness of deciding~(i)
implies the \ETR-hardness of deciding any of the other statements.

Indeed, according to Chaplick et al.~\cite[proof of
Theorem~1]{ChaplickFLRVW16b},
$G$ is an arrangement graph if and only if $\rho^1_2(G')\leq k$, that
is, (i) and (ii)~are equivalent.

Assume~(i). If~$G$ corresponds to a line arrangement of~$k$ lines,
all edges of $G$ lie on these $k$ lines and the tails of $G'$ can be
added without increasing the number of lines. This arrangement shows
that $\segtwo(G')\le k$, that is, (i) implies~(iii).

Assume~(iii), i.e., $\segtwo(G')\le k$.  Then $\segbend(G')\le k$~(iv) and
$\segx(G')\le k$~(v).

Assume~(iv), i.e., $\segbend(G')\le k$.
Let~$\Gamma'$ be a polyline drawing of~$G'$ on $\segbend(G')$ segments.
The graph~$G'$ contains $\binom{k}{2}$ degree-4 vertices.
As each of these vertices lies on the intersection of two
segments in~$\Gamma'$, we need $k$ segments to get enough
intersections, that is, $\segbend(G')\ge k$.
Thus $\segbend(G')=k$
and each intersection of the segments
of $\Gamma'$ (in particular, each bend) is a vertex of~$G'$.
Therefore edges in $\Gamma'$ do not bend in interior points and
$\Gamma'$ witnesses that $\segtwo(G)\le k$.  Thus~(iv) implies~(ii).

Finally, assume~(v), i.e., $\segx(G')\le k$.  Let~$\Gamma$ be a
straight-line drawing with possible crossings on $\segx(G')$ segments.  Again, we
need $k$ segments to get enough intersections, that is,
$\segx(G')\ge k$.  Thus $\segx(G')=k$ and
each intersection of the segments of $\Gamma'$ is a vertex of~$G'$.
Therefore edges in $\Gamma'$ do not cross and  $\Gamma'$
witnesses that $\segtwo(G)\le k$.  Thus~(v) implies~(ii).

Summing up, (iii) implies (iv) and (v), which both imply (ii), which
implies~(i), which implies~(iii).  Hence, all statements are equivalent.
\end{proof}

\begin{theorem}
  \label{thm:three-hard}
  Given a graph $G$ of maximum degree $4$ and an integer $k$, it is \ETR-hard to decide
  whether $\segthree(G) \le k$.
\end{theorem}

\begin{proof}
  Chaplick et al.~\cite[proof of Theorem~1]{ChaplickFLRVW16b} argued
  that for the graph~$G'$ constructed in the proof of
  Theorem~\ref{thm:hard} above, it holds that
  $\rho^1_2(G')=\rho^1_3(G')$.  Then, by the proof of
  Theorem~\ref{thm:hard}, we have $\rho^1_3(G') = \segx(G')$.

  By definition, we immediately obtain
  $\segthree(G') \le \rho^1_3(G')$.  By
  Corollary~\ref{cor:projection}, we have that
  $\segx(G') \le \segthree(G')$.  Therefore,
  $\segx(G') = \segthree(G')$.  Together with the arguments in the
  proof of Theorem~\ref{thm:hard}, this implies the theorem.
\end{proof}

\begin{theorem}
  \label{thm:complete}
  Given a planar graph $G$ and an integer $k$, it is
  \ETR-complete to decide whether $\segtwo(G) \le k$, whether
  $\segthree(G) \le k$, whether $\segbend(G) \le k$, and whether
  $\segx(G) \le k$.
\end{theorem}

\begin{proof}
  Given the hardness results in Theorems~\ref{thm:hard}
  and~\ref{thm:three-hard}, it remains to show that each of the four
  problems lies in \ETR.  Chaplick et al.~\cite[ArXiv version,
  Section~2]{ChaplickFLRVW16b} have shown that deciding whether
  $\rho_1^2(G) \le k$ and $\rho_1^3(G) \le k$ both lie in \ETR.  To
  this end, they showed that these questions can be formulated as
  first-order existential expressions over the reals.  We now show how
  to extend their expression for deciding whether~$\rho_1^2(G) \le k$
  to an expression for deciding whether $\segtwo(G) \le k$.  The
  expressions for the other variants can be extended in a similar way.

  Their existential statement over the reals starts with the
  quantifier prefix $\E v_1\ldots\E v_n\E p_1\E q_1\ldots \E p_k\E
  q_k$, where quantification $\E a$ over a point $a=(x,y)$ means the
  quantifier block $\E x\E y$, the points $v_1,\dots,v_n$ are the
  points to which the vertices of~$G$, $\{1,\dots,n\}$, are mapped,
  and the pairs $(p_1,q_1) \dots, (p_k,q_k)$ define the $k$ lines that
  cover the drawing of~$G$.  The expression~$\Pi$ over which they
  quantify uses a subexpression that takes as input three points in
  $\mathbb{R}^2$; for $a$, $b$, and~$c$, they define the expression
  $B(a,b,c)$ such that it is true if and only if~$a$ lies on the line
  segment~$\overline{bc}$.

  To the expression~$\Pi$ we simply add a term that ensures that, for
  each pair of consecutive points $v_i$ and $v_j$ on the same line,
  vertices~$i$ and~$j$ are adjacent in~$G$:
  \[\bigwedge_{l \in \{1,\dots,k\}, i,j,k \in V} B(v_i,p_l,q_l)
  \wedge B(v_j,p_l,q_l) \wedge \neg B(v_k,v_i,v_j) \Rightarrow \{i,j\}
  \in E,\]
  where $V$ is the vertex set and $E$ is the edge set of the graph $G$. 
\end{proof}

\section{Algorithms and Lower Bounds for Cubic Graphs}
\label{sec:cubic}

Consider a polyline drawing~$\delta$ of a cubic graph (in 2D or
3D).  Note that there are two types of vertices; those where exactly
one segment ends and those where three segments end.  We call these
vertices \emph{flat vertices} and \emph{tripods}, respectively.  Let
$f(\delta)$ be the number of flat vertices, $t(\delta)$ the
number of tripods, and $b(\delta)$ the number of bends in~$\delta$.

\begin{lemma}
  \label{lem:flat}
  For any straight-line drawing~$\delta$ of a cubic graph with $n$
  vertices, $\seg(\delta)=3n/2-f(\delta)+b(\delta)=n/2+t(\delta)+b(\delta)$.
\end{lemma}

\begin{proof}
  Clearly, $n=f(\delta)+t(\delta)$.  The number of ``segment ends'' is
  $3t(\delta)+f(\delta)+2b(\delta)=3n-2f(\delta)+2b(\delta)=n+2t(\delta)+2b(\delta)$.  The claim
  follows since every segment has two ends.
\end{proof}

\subsection{Singly-Connected Cubic Graphs}

\begin{proposition}
  \label{prop:5n-over-6-example}
  There is an infinite family $(G_k)_{k \ge 1}$ of connected cubic
  graphs such that $G_k$ has $n_k=6k-2$ vertices and
  $\segtwo(G_k) = \segthree(G_k) = \segbend(G_k) = \segx(G_k) = 5k-1 =
  5n_k/6+2/3$.
\end{proposition}

\begin{proof}
  Let $K_4'$ be the graph $K_4$ with a subdivided edge.  Consider the
  graph~$G_k$ depicted in Fig.~\ref{fig:5n-over-6-example} (for
  $k=4$).  It consists of a caterpillar with $k-2$ inner vertices (of
  degree~3) where each of the $k$ leaf nodes is replaced by a copy
  of~$K_4'$.  The convex hull of every polyline drawing of~$K_4'$ has
  at least three extreme points.  One of these points may connect~$K_4'$
  to~$G_k-K_4'$, but each of the remaining two must be a tripod or a bend.
  This holds for every copy of~$K_4'$.
  Hence, for any drawing~$\delta$ of~$G$, $t(\delta)+b(\delta)\ge 2k$.  Now
  Lemma~\ref{lem:flat} yields that $\seg(\delta) \ge 5k-1$.  For the
  drawing in Fig.~\ref{fig:5n-over-6-example}, the bound is tight.
\end{proof}

\begin{figure}[tb]
  \centering
  \includegraphics{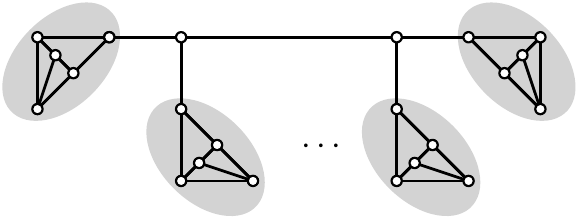}
  \caption{The graph $G_k$ (here $k=4$) is a caterpillar with $k-2$
    inner vertices of degree~3 where each leaf has been replaced by a
    copy of the 5-vertex graph~$K_4'$ (shaded gray).}
  \label{fig:5n-over-6-example}
\end{figure}

\subsection{Biconnected Cubic Graphs}

\begin{proposition}
  \label{prop:K33-cycle}
  There is an infinite family of Hamiltonian (and hence biconnected)
  cubic graphs $(H_k)_{k \ge 3}$ such that $H_k$ has $n_k=6k$
  vertices, $\segthree(H_k)=5k=5n_k/6$, and $\segx(H_k)=4k=2n_k/3$.
\end{proposition}

\begin{proof}
  Consider the graph~$H_k$ depicted in Fig.~\ref{fig:K33-cycle} (for
  $k=4$).  It is a $k$-cycle where each vertex is replaced by a copy
  of a 6-vertex graph~$K$ ($K_{3,3}$ minus an edge).  The graph~$H_k$
  has $n_k=6k$ vertices and is not planar.

  In any 2D drawing of the subgraph~$K$, at least three vertices lie
  on the convex hull of the drawing of~$K$.  Two of these vertices may
  connect $K$ to $H_k-K$, but at least one of the convex-hull vertices
  is a tripod.  This holds for every copy of~$K$.  Hence, for any
  (3D) drawing~$\delta$ of~$H_k$, $t(\delta) \ge k$.  Now
  Lemma~\ref{lem:flat} yields that $\seg(\delta) \ge n_k/2+k=2n_k/3$.
  The same bound holds for~$\segx(H_k)$.

  In order to bound $\segthree(H_k)$ we consider two possibilities for
  the drawing of the subgraph~$K$; either it lies in a plane or it
  doesn't.  In the planar case, the two vertices that connect~$K$ to
  $H_k-K$ cannot lie in the same face of the planar embedding of~$K$
  (otherwise we could connect these two vertices without crossings,
  contradicting the fact that $K_{3,3}$ is not
  planar).  Hence, at least two vertices on the convex hull of~$K$
  must be tripods.  In the non-planar case, the convex hull consists
  of four vertices.  Two of these may connect~$K$ to $H_k-K$, but
  again at least two must be tripods.  In both cases we hence have
  $t(\delta) \ge 2k$ for any 3D drawing~$\delta$ of~$H_k$.
  Now Lemma~\ref{lem:flat} yields $\seg(\delta) \ge n_k/2+2k=5n_k/6$.
  The same bound holds for~$\segthree(H_k)$.

  For the drawing in Fig.~\ref{fig:K33-cycle}, the bound for \segx is
  tight.  Lifting the $k$ white vertices that do not lie on the outer
  face from the xy-plane ($z=0$) to the plane $z=1$, yields a
  crossing-free 3D drawing where the bound for \segthree is tight.
\end{proof}

\begin{figure}[tb]
  \begin{minipage}[b]{.48\textwidth}
    \centering
    \includegraphics{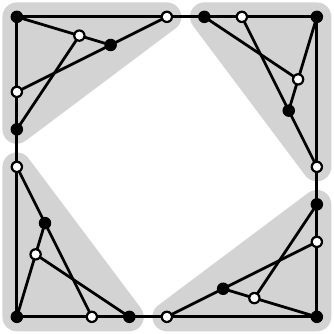}
    \caption{The cubic graph $H_k$ (here $k=4$) is a $k$-cycle whose
      vertices are replaced by $K_{3,3}$ minus an edge (shaded).
    }
    \label{fig:K33-cycle}
  \end{minipage}
  \hfill
  \begin{minipage}[b]{.48\textwidth}
    \centering
    \includegraphics{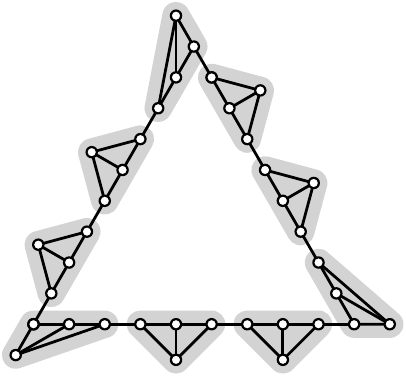}
    \caption{The planar cubic graph~$I_k$ (here $k=9$) is a $k$-cycle
      whose vertices are replaced by~$K_4$ minus an edge (shaded).
    }
    \label{fig:claw-cycle}
  \end{minipage}
\end{figure}

\begin{proposition}
  \label{prop:claw-cycle}
  There is an infinite family of planar cubic Hamiltonian (and hence
  biconnected) graphs
  $(I_k)_{k \ge 3}$ such that~$I_k$ has $n_k=4k$ vertices and
  $\segtwo(I_k)=\segthree(I_k)=\segbend(I_k)=\segx(I_k)=3k=3n_k/4$.
\end{proposition}

\begin{proof}
  Consider the graph~$I_k$ depicted in Fig.~\ref{fig:claw-cycle} (for
  $k=9$).  It is a $k$-cycle where each vertex is replaced by a copy
  of the graph~$K'$, which is $K_4$ minus an edge.  Therefore, $I_k$
  has $4k$ vertices.  The depicted drawing consists of~$3k$ segments.
  This yields the upper bounds.

  Concerning the lower bounds, note that, in any drawing style, each
  subgraph~$K'$ has an extreme point %
  not connected to $I_k-V(K')$.  This %
  point must be a tripod or a bend.  Hence, in any drawing~$\delta$
  of~$I_k$, $t(\delta) + b(\delta) \ge k$ and, by
  Lemma~\ref{lem:flat}, $\segtwo(I_k) = \segthree(I_k) = \segbend(I_k)
  = \segx(I_k) \ge 2k+t(\delta)+b(\delta) \ge 3k$.
\end{proof}

\begin{theorem}
  \label{thm:planar-bi-cubic-bend}
  For any biconnected planar cubic graph~$G$ with $n$ vertices, it
  holds that $\segbend(G) \le n+1$.  A corresponding drawing can be
  found in linear time.
\end{theorem}

\begin{proof}
  We draw~$G$ using the algorithm of Liu et al.~\cite{lmp-sbecr-AM94}
  that draws any planar biconnected cubic graph except the tetrahedron
  orthogonally with at most one bend per edge and at most $n/2+1$
  bends in total.  It remains to count the number of segments in this
  drawing.  In any vertex exactly one segment ends; in any bend
  exactly two segments end.  In total, this yields at most $n+2 \cdot
  (n/2+1)=2n+2$ segment ends and at most $n+1$ segments.

  Concerning the special case of the tetrahedron ($K_4$), note that it
  can be drawn with five segments when bending one of its six edges.
\end{proof}

\begin{open}
  What about 4-regular graphs?  They have $2n$ edges.  If we bend
  every edge once, we already need $2n$ segments~-- and not all
  4-regular graphs can be drawn with at most one bend per edge.
\end{open}

Every biconnected graph~$G$ admits an \emph{st-numbering},
that is, an ordering $\langle v_1,\dots,v_n \rangle$ of the vertex set
$\{v_1,\dots,v_n\}$ of~$G$ such that for every $j \in \{2,\dots,n-1\}$
vertex~$v_j$ has at least one predecessor (that is, a neighbor~$v_i$
with $i<j$) and at least one successor (that is, a neighbor~$v_k$ with
$k>j$).  Such a numbering can be computed in linear
time~\cite{et-cstn-TCS76}.
Given a cubic graph with an st-numbering
$\langle v_1,\dots,v_n \rangle$, we call a
vertex~$v_j$ with $j \in \{1,\dots,n\}$ a \emph{$p$-vertex} if it has
$p$ predecessors; $p \in \{0,1,2,3\}$.

\begin{lemma}
  \label{lem:12-vertices}
  Given a biconnected cubic graph with an st-numbering $\langle
  v_1,\dots,v_n \rangle$, there is one 0-vertex and one 3-vertex and
  there are $(n-2)/2$ 1-vertices and $(n-2)/2$ 2-vertices.
\end{lemma}

\begin{proof}
  Direct every edge from the vertex with smaller index to the vertex
  with higher index.  In the resulting directed graph, the sum of the
  indegrees equals the sum of the outdegrees.  Hence, the number of
  1-vertices (with indegree~1 and outdegree~2) and the number of
  2-vertices (with indegree~2 and outdegree~1) must be equal.
  It is obvious that there is one 0- and 3-vertex each.
\end{proof}

\begin{theorem}
  \label{thm:bi-cubic-3d}
  For any biconnected cubic graph~$G$ with $n$ vertices, $\segthree(G)
  \le n+2$ and $\segx(G) \le n+2$.
\end{theorem}
\begin{proof}
  We show that $\segthree(G) \le n+2$.  Then
  Corollary~\ref{cor:projection} yields $\segx(G) \le n+2$.  For two
  different points $x$ and $y$ in $\mathbb{R}^3$, we denote the line
  that goes through~$x$ and~$y$ by~$xy$.

  Let $\langle v_1,\dots,v_n \rangle$ be an st-numbering of~$G$.  We
  construct a drawing~$\delta$ of~$G$, going through the vertices
  according to the st-numbering and using x-coordinate~$j \pm
  \varepsilon$ for vertex~$v_j$, where $0<\varepsilon \ll 1$.  We
  place~$v_1$ at $(1,1,1)$.  At every step $j=2,\dots,n$, we maintain a
  set~$\mathcal{L}$ of lines that are directed to the right such that any
  two lines in~$\mathcal{L}$ are either skew (that is, they don't lie
  in the same plane) or they intersect and their unique intersection
  point is the location of a vertex~$v_k$ with $k \le j$ (that is, the
  intersection point is~$v_j$ or it lies to the left of~$v_j$).
  Initially, $\mathcal{L}$ is empty.

  If~$v_j$ is a 1-vertex, we differentiate two cases depending on the
  unique predecessor~$v_i$ of~$v_j$.

  \smallskip

  \noindent
  \textbf{Case~I:} If~$v_i$ is the last vertex on a line~$\ell$
  in~$\mathcal{L}$, we place~$v_j$ on the intersection point of~$\ell$
  with the plane $x=j$.  In this case, the set~$\mathcal{L}$ doesn't
  change.

  \smallskip

  \noindent
  \textbf{Case~II:} Otherwise, we place~$v_j$ in the plane $x=j$
  such that the line~$v_iv_j$ is
  skew with respect to all lines in~$\mathcal{L}$ except for the
  line~$\ell$ that contains~$v_i$ and the unique predecessor of~$v_i$.
  (Note that the predecessor of~$v_i$ and the line~$\ell$ don't exist
  if $i=1$.)
  Clearly, $v_iv_j$ and~$\ell$ intersect in~$v_i$ and $i<j$.  Hence,
  we can add the line $v_iv_j$ to the set~$\mathcal{L}$.

  \smallskip

  If~$v_j$ is a 2-vertex, let~$v_i$ and~$v_{i'}$ be the two
  predecessors of~$v_j$.  Again, we consider two cases.

  \smallskip

  \noindent
  \textbf{Case~I':} At least one of $v_i$ or $v_{i'}$ is flat (that is, it lies
  on an inner point of the segment created by its incident edges that
  have already been drawn) or one of them is the vertex $v_1$.

  In this case, we treat~$v_j$ similarly as in Case~II
  above; we make sure that the lines~$v_iv_j$ and~$v_{i'}v_j$ are skew
  with respect to all lines in~$\mathcal{L}$ except that~$v_iv_j$
  won't be skew with respect to the at most two lines that
  connect~$v_i$ to its predecessors and~$v_{i'}v_j$ won't
  be skew with respect to the at most two lines that connect~$v_{i'}$
  to its predecessors.  Note that~$v_iv_j$ intersects any
  line through~$v_i$ and its neighbors in~$v_i$, and it holds that $i<j$.
  Similarly, $v_{i'}v_j$ intersects any line through~$v_{i'}$ and its
  neighbors in~$v_{i'}$, and it holds that $i'<j$.  The lines $v_iv_j$
  and~$v_{i'}v_j$ intersect in~$v_j$.  Hence, we can add the
  lines~$v_iv_j$ and~$v_{i'}v_j$ to the set~$\mathcal{L}$.

  \smallskip

  \noindent
  \textbf{Case~II':} Both $v_i$ and $v_{i'}$ are the last vertices on their
  lines~$\ell$ and~$\ell'$, respectively.

  If one of them, say $v_i$, has a successor~$v_k$ with $k>j$, we
  extend the line~$\ell$ of~$v_i$ and put $v_j$ on the intersection
  of~$\ell$ and the plane~$x=j$.

  Otherwise $v_i$ has a successor $v_k$ with $k<j$ and $v_{i'}$ has a
  successor $v_{k'}$ with $k'<j$, which both don't lie on the
  lines~$\ell$ and~$\ell'$.
  In this case, we put~$v_j$ on one of $\ell$ and $\ell'$, say~$\ell$,
  and add the line $v_{i'}v_j$ to the set $\mathcal{L}$.  Now we pick
  some $0<\varepsilon \ll 1$ such that we can place~$v_j$ at the
  intersection of~$\ell$ and $x=j+\varepsilon$.  We must avoid to
  place~$v_j$ on a plane spanned by any two non-skew lines
  in~$\mathcal L$ (intersecting to the left of $x=j$).  With this
  trick, the invariant for $\mathcal L$ still holds since the new line
  in~$\mathcal L$, $v_{i'}v_j$, intersects only~$\ell'$ (in $v_{i'}$,
  hence to the left).

  \smallskip

  Finally, we place~$v_n$ (which is a 3-vertex)
  at a point in the plane $x=n$ that does not lie on any of the lines
  spanned by pairs and planes spanned by triples of previously placed vertices.

  This finishes the description of the drawing~$\delta$ of~$G$.  Due to
  our invariant regarding the set~$\mathcal{L}$, no two edges
  of~$G$ intersect in~$\delta$.

  To bound the number of segments in~$\delta$, we use a simple
  charging argument.  Each non-first and non-last vertex~$v$ has a predecessor
  which is a flat vertex or~$v_1$. To this predecessor~$v$ pays a coin.  On
  the other hand, $v_1$ receives at most three coins and every flat
  vertex receives at most two coins.  Hence, $f(\delta) \ge (n-5)/2$.
  Since $n$ is even, $f(\delta) \ge n/2-2$.  Now, Lemma~\ref{lem:flat}
  yields the claim.
\end{proof}

\subsection{Triconnected Cubic Graphs}
\label{sub:triconnected}

\begin{proposition}
  \label{prop:tri-cubic-3d}
  There is an infinite family of triconnected cubic graphs $(F_k)_{k
    \ge 4}$ such that~$F_k$ has $n_k=5k$ vertices and
  $\segthree(F_k)=3.5k=7n_k/10$.
\end{proposition}

\medskip
\noindent\emph{Proof.}~
  Let~$G_k$ be an arbitrary triconnected cubic graph with $k$ vertices
  ($k$ even).  By Steinitz's theorem,
  there exists a drawing of the graph $G_k$ as a 1-skeleton of a 3D
  convex polyhedron. Replace each vertex~$v$ of~$G_k$ by a copy
  of~$K_{2,3}$ as shown in Fig.~\ref{fig:triconnected}, where $v$ is the
  central (orange) vertex---a tripod---, all
  other vertices of the copy are flat, and the three arrows correspond
  to the three edges of~$G_k$.
  The resulting geometric graph~$F_k$ has $n_k=5k$ vertices and is not
  planar.  Since~$F_k$ has $k$ tripod vertices, by
  Lemma~\ref{lem:flat}, $\segthree(F_k)\le n_k/2+k=3.5k=7n_k/10$.

  \begin{wrapfigure}[9]{r}{.26\textwidth}
    \centering
    \vspace{-4ex}
    \includegraphics{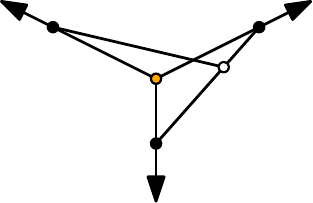}
    \caption{Gadget for the proof of
      Proposition~\ref{prop:tri-cubic-3d}}
    \label{fig:triconnected}
  \end{wrapfigure}

  In order to bound $\segthree(F_k)$ from below, we consider two
  possibilities for the drawing of each subgraph~$K_{2,3}$; either it lies
  in a plane or it doesn't. In the planar case, the convex hull of the drawing
  has at least three extreme points. If none of them was a tripod then
  there would be exactly three extreme points, each a black vertex.
  Thus we could place an additional white vertex in the exterior
  of the convex hull and connect it to all black vertices, obtaining
  an impossible plane drawing of~$K_{3,3}$.
  In the non-planar case, the convex hull consists of at least four
  vertices.  Three of these may connect~$K_{2,3}$ to $F_k-V(K_{2,3})$,
  but again at least one must be a tripod.

  In both cases we hence have $t(\delta) \ge k$ for any 3D
  drawing~$\delta$ of~$F_k$.
  Now Lemma~\ref{lem:flat} yields $\seg(\delta) = n_k/2+t(\delta) \ge 3.5k$.
\hfill\qed

\section{Open Problems}
\label{sec:open}

Apart from improving our bounds, we have the following open problem.

\begin{open}
  Can we produce drawings in~3D (or with bends or crossings in~2D)
  that fit on grids of small size?
\end{open}

\paragraph{Acknowledgments.}

We thank the organizers and participants of the 2019 Dagstuhl seminar
``Beyond-planar graphs: Combinatorics, Models and Algorithms''.  In
particular, we thank G\"unter Rote and Martin Gronemann for
suggestions that led to some of this research.  We also thank Carlos
Alegr\'{i}a.  We thank our reviewers for an idea that improved the
bound in Proposition~\ref{prop:tri-cubic-3d}, for suggesting the
statement of Lemma~\ref{lem:projection}, and for many other helpful
comments.

\bibliographystyle{splncs04}
\bibliography{abbrv,references}

\end{document}